\documentclass[conference,letterpaper]{IEEEtran}
\IEEEoverridecommandlockouts
\usepackage{cite}
\usepackage{amsmath,amssymb,amsfonts,amsthm}
\usepackage{algpseudocode}
\usepackage{graphicx}
\usepackage{textcomp}
\usepackage{xcolor}
\usepackage{booktabs}
\usepackage{algorithm}
\usepackage{enumerate}

\def\BibTeX{{\rm B\kern-.05em{\sc i\kern-.025em b}\kern-.08em
    T\kern-.1667em\lower.7ex\hbox{E}\kern-.125emX}}

\newcommand{\UAV}{UAV}
\newcommand{\GBS}{GBS}
\newcommand{\GU}{GU}
\newtheorem{proposition}{Proposition}

\begin{document}

\title{Learn to Access and Backhaul the Sky: Multi-Scale Radio Map Guided Multi-UAV Cooperation}


\author{\IEEEauthorblockN{Yifeng Yuan, and Shijian Gao}
\IEEEauthorblockA{
IoT Thrust, The Hong Kong University of Science and Technology (Guangzhou), China
}
}

\maketitle

\begin{abstract}
Driven by the emerging low-altitude economy, uncrewed aerial vehicle (UAV) swarms offer flexible integrated air-ground access and backhaul. However, providing seamless connectivity is difficult due to the interdependent dynamics of user mobility and building blockages in these 3D scenarios. These factors create rapidly shifting bottlenecks in end-to-end paths. Furthermore, the multi-dimensional nature of joint control limits the effectiveness of traditional heuristics.
To address these challenges, a \textbf{\underline{M}}ulti-Scale \textbf{\underline{R}}adio \textbf{\underline{M}}ap-\textbf{\underline{G}}uided (MRMG) framework is proposed. The MRMG framework handles heterogeneous dynamics by integrating three distinct levels of radio information: global-level maps provide regional coverage insights, local-level maps capture neighborhood-scale service conditions, and link-level maps characterize high-resolution channel features. This design effectively decouples macro-movement from micro-link adaptation. To yield long-term performance improvements, A multi-agent reinforcement learning (MARL) controller learns cooperative policies for UAV movement, next-hop selection, and transmit-power control. Simulation results show that the MRMG framework not only improves network throughput but also significantly bolsters cell-edge service, nearly doubling the 5th-percentile user rate.
\end{abstract}


\section{Introduction}
\label{sec:intro}

With the rapid expansion of the low-altitude economy, air-ground integrated networks have emerged as a pivotal architecture for future wireless evolution~\cite{9112744}. Uncrewed aerial vehicle (UAV) swarms provide vital aerial access and relaying support, bridging the gap between terrestrial infrastructure and dynamic user demand~\cite{gao2026integrated, liu2024prospective}. However, in 3D air-ground scenarios, UAVs must balance a dual role: providing high-quality ground-user access while maintaining reliable backhaul links to ground base stations (GBSs)~\cite{LLLK_CCHI_2025}. This creates a complex spatial coupling between access, backhaul, and terrestrial interference. A location optimized for user coverage often suffers from severe building blockage toward the backhaul GBS, while aerial transmissions may simultaneously degrade ground-side reception through line-of-sight (LoS) interference.

Existing research has attempted to address these challenges through either offline optimization or online learning. Studies focused on access-side optimization alone~\cite{cai2025multi} often ignore the end-to-end (E2E) bottleneck created by the backhaul link. Offline placement studies that consider integrated access and backhaul rely on user distributions and node coordinates to optimize UAV positions and backhaul topology~\cite{sabzehali2022optimizing, mahmood2023joint, zhang2023deployment}. Online learning approaches further account for dynamic UAV movement, routing, and transmit power control~\cite{wang2023deep, wang2025dynamic, sheng2024enabling, ding2022packet, wang2022learning, alam2024joint}. However, both categories of methods are insufficient in urban 3D environments, where building blockages cause extreme spatial variations in signal propagation that coordinates alone cannot capture. While radio map construction for low-altitude UAV networks has attracted recent interest~\cite{lu2026transfer, gao2026farm}, their use in UAV placement and routing to provide environment-aware propagation priors has also been explored~\cite{li2022derivative, li2024radio, li2025radio}. Yet existing solutions remain largely limited to static deployment or predictive planning under fixed trajectories. A significant gap remains in developing an effective framework that incorporates multi-scale propagation awareness into online, cooperative UAV swarm decision-making.

To fill this gap, this paper proposes a Multi-Scale Radio Map-Guided (MRMG) framework for online cooperative UAV operation. The framework transforms raw radio information into three complementary scales to handle the heterogeneous dynamics of 3D air-ground networks. First, global-level maps provide regional coverage insights for macro-path planning. Second, local-level maps capture neighborhood-scale service conditions, and third, link-level maps characterize transient backhaul feasibility. These radio-aware observations are then incorporated into a MARL controller based on multi-agent proximal policy optimization (MAPPO), enabling the UAV swarm to learn long-term cooperative policies. The learned policy jointly determines UAV movement, next-hop selection, and transmit-power control during online operation. Given the resulting aerial backhaul topology, a topology-aware flow allocation step then computes feasible E2E user rates under access and multi-hop backhaul constraints. Consequently, this approach effectively decouples macro-movement from micro-link adaptation. Simulations demonstrate the effectiveness of the proposed MRMG framework in complex urban scenarios. Compared to state-of-the-art baselines, MRMG improves service coverage by 6.8 percentage points and significantly bolsters cell-edge performance, nearly doubling the 5th-percentile user rate. 
These gains confirm that the multi-scale integration of radio maps is essential for maintaining robust E2E connectivity in highly dynamic and obstructed 3D environments.

\section{System Modeling and Problem Formulation}
\label{sec:system}

We consider an air-ground network comprising a swarm of $M$ UAVs, $N_{\mathrm{B}}$ ground base stations (GBSs), and $K$ ground users (GUs). At slot $t$, GBS $b$ is fixed at $\mathbf{g}_b$, while UAV $m$ and GU $k$ are positioned at $\mathbf{p}_m(t)$ and $\mathbf{u}_k(t)$, respectively. GUs are served either by a GBS or via a UAV-aided relay path, with traffic forwarded to GBS through aerial backhaul. The disjoint sets of users served by UAV $m$ and GBS $b$ are denoted as $\mathcal{K}_m(t)$ and $\mathcal{K}_b(t)$. In each slot, every UAV jointly optimizes its displacement, next-hop selection, and transmit power to balance access quality, backhaul feasibility, and air-to-ground interference.

\begin{figure}[t]
    \centering
    \includegraphics[width=0.95\columnwidth]{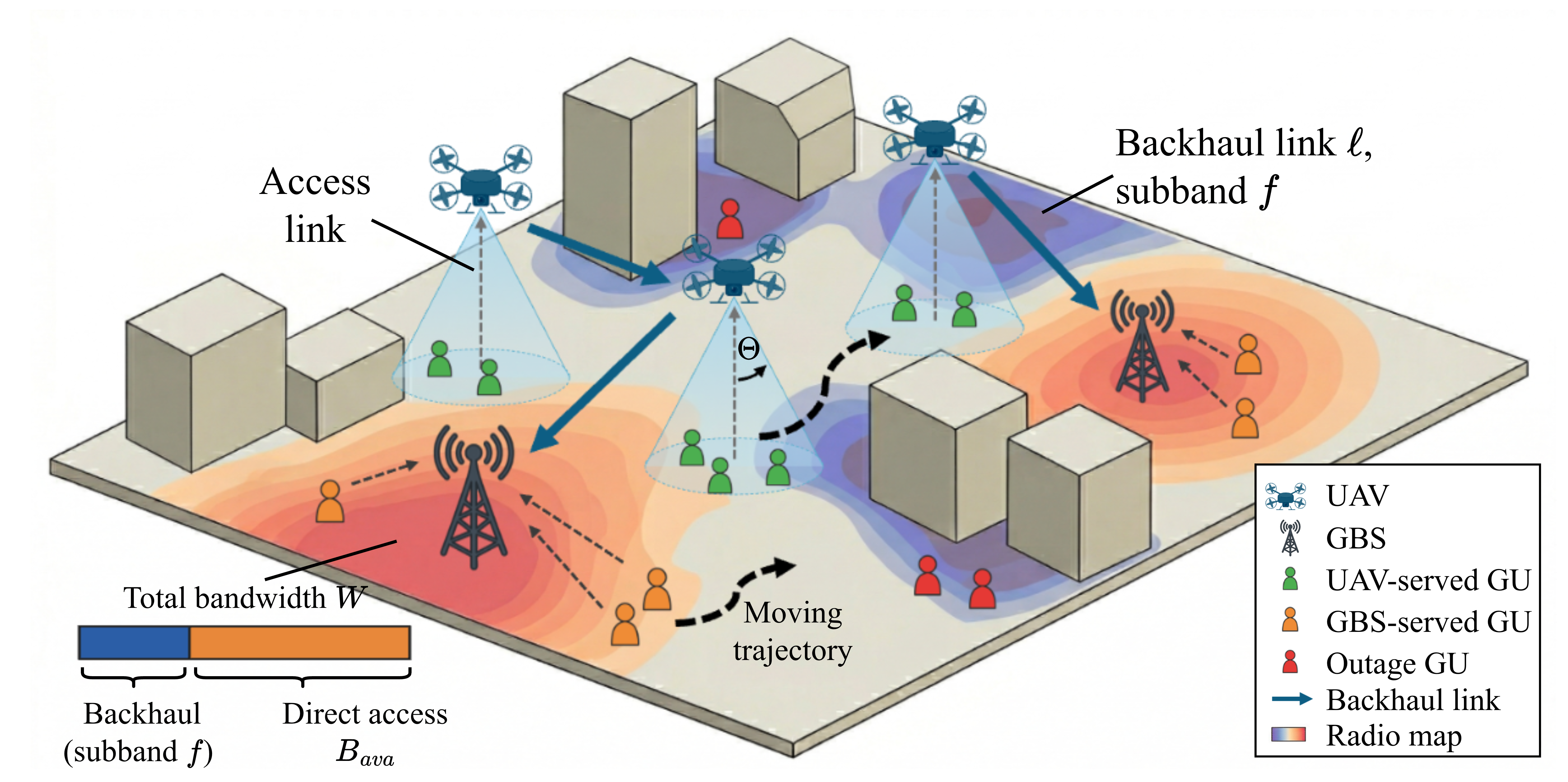}
    \caption{The considered air-ground network: UAVs provide uplink access to ground users and backhaul traffic to GBSs via multi-hop aerial links.}
    \label{fig:scenario}
\end{figure}

\subsection{Link Model}

We model three categories of links: \GU{}--\UAV{} access links, aerial backhaul links, and \GU{}--\GBS{} direct links.
The total uplink bandwidth $W$ is divided into $F$ orthogonal subbands $\mathcal{F}=\{1,\dots,F\}$, each of bandwidth $B_{\mathrm{sub}}=W/F$.
Each \UAV{} operates on one assigned subband; backhaul links on different subbands do not interfere, though aerial transmissions may act as unintended A2G interferers at non-intended \GBS{} receivers.
All large-scale channel gains are read from precomputed radio maps (Section~\ref{subsec:radiomap}), except \UAV{}--\UAV{} links, which use a distance-dependent air-to-air model.

\subsubsection{\textbf{\GU{}--\UAV{} Access Links}}

User~$k$ can associate with \UAV{}~$m$ only when it falls within the coverage cone,
\begin{equation}
    \arctan\!\left(\frac{\|\mathbf{u}_k^\perp(t)-\mathbf{p}_m^\perp(t)\|_2}{p_{m,z}(t)}\right)\leq\Theta,
\end{equation}
where $\mathbf{u}_k^\perp$ and $\mathbf{p}_m^\perp$ are horizontal coordinates, $p_{m,z}$ is the \UAV{} altitude, and $\Theta$ is the half-beam angle.
Since each \UAV{} occupies a distinct subband, inter-\UAV{} access interference is absent.
For $k\in\mathcal{K}_m(t)$, bandwidth is equally shared, giving $B_{k,m}^{\mathrm{acc}}(t)={B_{\mathrm{sub}}}/{|\mathcal{K}_m(t)|}$ and the uplink access rate
\begin{equation}
R_{k,m}^{\mathrm{acc}}(t)=
  B_{k,m}^{\mathrm{acc}}(t)
  \log_2\!\left(
    1+\frac{P_{\mathrm{GU}}\,g_{k,m}(t)}
           {N_0\,B_{k,m}^{\mathrm{acc}}(t)}
  \right),
\label{eq:uav_acc_rate}
\end{equation}
where $P_{\mathrm{GU}}$ is \GU{} transmit power, $g_{k,m}(t)$ is the UAV-to-ground channel gain, and $N_0$ is the noise power spectral density.

\subsubsection{\textbf{Aerial Backhaul Links}}

Each active backhaul link $\ell=(m,n)\in\mathcal{L}(t)$ carries transmissions from node~$m$ to its selected next hop $n\in(\mathcal{M}\cup\mathcal{B})\setminus\{m\}$ over a dedicated subband, giving link capacity
\begin{equation}
C_\ell(t)=
  B_{\mathrm{sub}}
  \log_2\!\left(
    1+\frac{P_m(t)\,g_{m,n}(t)}
           {N_0\,B_{\mathrm{sub}}}
  \right),
\label{eq:backhaul_cap}
\end{equation}
where $g_{m,n}(t)$ is the GBS-to-UAV channel gain when $n$ is a \GBS{}, or the air-to-air gain when $n$ is another \UAV{}.
The aggregate flow on each link is bounded by its capacity:
\begin{equation}
\sum_{j:\,\ell\in\mathcal{P}_j(t)} r_j(t)
  \leq C_\ell(t),
  \quad\forall\ell\in\mathcal{L}(t),
\label{eq:backhaul_flow}
\end{equation}
where $\mathcal{P}_k(t)$ is the backhaul path of user~$k$ and $r_j(t)$ is the allocated end-to-end rate.

\subsubsection{\textbf{\GU{}--\GBS{} Direct Links}}

At \GBS{}~$b$, the $N_b^{\mathrm{occ}}(t)$ terminating backhaul links each occupy one subband, leaving residual bandwidth $B_b^{\mathrm{ava}}(t)=W-N_b^{\mathrm{occ}}(t)B_{\mathrm{sub}}$ for the $|\mathcal{K}_b(t)|$ direct users, so $B_{k,b}^{\mathrm{acc}}(t)={B_b^{\mathrm{ava}}(t)}/{|\mathcal{K}_b(t)|}$.
Aerial transmissions directed to other next hops impose aggregate interference
\begin{equation}
I_b^{\mathrm{UAV}}(t)=
  \sum_{m\in\mathcal{M}:\,n_m(t)\neq b}
    P_m(t)\,g_{m,b}(t),
\label{eq:gbs_interference}
\end{equation}
and the direct-access rate of user~$k\in\mathcal{K}_b(t)$ is
\begin{equation}
R_{k,b}^{\mathrm{acc}}(t)=
  B_{k,b}^{\mathrm{acc}}(t)
  \log_2\!\left(
    1+\frac{P_{\mathrm{GU}}\,g_{k,b}(t)}
           {N_0\,B_{k,b}^{\mathrm{acc}}(t)+I_b^{\mathrm{UAV}}(t)}
  \right),
\label{eq:gbs_acc_rate}
\end{equation}
where $g_{k,b}(t)$ and $g_{m,b}(t)$ are the GBS-to-ground and GBS-to-UAV channel gains, respectively.

\subsection{Problem Formulation}

We aim to maximize a weighted long-term utility that captures both aggregate throughput and service coverage. Depending on whether user~$k$ is served directly by a \GBS{} or through the aerial-relay path, the delivered rate is
\begin{equation}
\widehat{r}_k(t)=
\begin{cases}
  R_{k,b}^{\mathrm{acc}}(t), & k\in\mathcal{K}_b(t),\\
  r_k(t),                    & k\in\mathcal{K}_m(t),
\end{cases}
\label{eq:delivered_rate}
\end{equation}
where $r_k$ is the allocated end-to-end rate for aerially served users. At slot $t$, the action of \UAV{} $m$ is
\begin{equation}
\mathbf{a}_m(t)=\big(\mathbf{d}_m(t),\,n_m(t),\,P_m(t)\big), \qquad m\in\mathcal{M},
\label{eq:joint_action_pf}
\end{equation}
where $\mathbf{d}_m(t)$, $n_m(t)$, and $P_m(t)$ denote the displacement, selected next hop, and transmit power, respectively. Let $\pi$ denote the joint control policy over all \UAV{}s. The optimization problem, denoted $\textbf{P0}$, is formulated as
\begin{subequations}\label{eq:P0}
\begin{alignat}{2}
  & \max_{\pi,\,\{r_k(t)\}} &\;
  & \mathbb{E}_{\pi}\!\left[\sum_{t=0}^{T-1}\!\left(
      \lambda_r\sum_{k\in\mathcal{K}} \hat{r}_k(t)
      +\lambda_c\sum_{k\in\mathcal{K}} \mathbb{I}\!\left(\hat{r}_k(t)\ge R_{\min}\right)
    \right)\right] \notag \\
  & \text{s.t.} &\;
  & \mathbf{d}_m(t)\in\mathcal{D}_m(t),
    \;\forall m\in\mathcal{M},
    \label{eq:pf_motion}\\
  & &\; & n_m(t)\in(\mathcal{M}\cup\mathcal{B})\setminus\{m\},
    \;\forall m\in\mathcal{M},
    \label{eq:pf_next_hop}\\
  & &\; & P_m(t)\in\mathcal{P},
    \;\forall m\in\mathcal{M},
    \label{eq:pf_power}\\
  & &\; & 0\le r_k(t)\le R_{k,m}^{\mathrm{acc}}(t),
    \;\forall k\in\mathcal{K}_m(t),\,m\in\mathcal{M},
    \label{eq:pf_access}\\
  & &\; & \sum_{j:\,\ell\in\mathcal{P}_j(t)} r_j(t)\le C_{\ell}(t),
    \;\forall \ell\in\mathcal{L}(t).
    \label{eq:pf_backhaul}
\end{alignat}
\end{subequations}
Here, $T$ is the episode horizon, $\mathcal{D}_m(t)$ is the feasible motion set of \UAV{} $m$, $\mathcal{P}$ is the discrete transmit-power set, $\lambda_r$ and $\lambda_c$ are nonnegative weights balancing throughput and coverage, $R_{\min}$ is the minimum target rate, and $\mathbb{I}(\cdot)$ is the indicator function. Constraint~\eqref{eq:pf_access} limits each aerially served user's rate to the access link capacity, and~\eqref{eq:pf_backhaul} enforces the backhaul capacity along every active link. Problem $\textbf{P0}$ is challenging due to the strong coupling among multi-\UAV{} control, location-dependent propagation, A2G interference, and multi-hop backhaul constraints.

\section{Radio-Map-Guided Learning for Cooperative UAV Control}
\label{sec:method}
\begin{figure*}[!t]
    \centering
    \includegraphics[width=\textwidth]{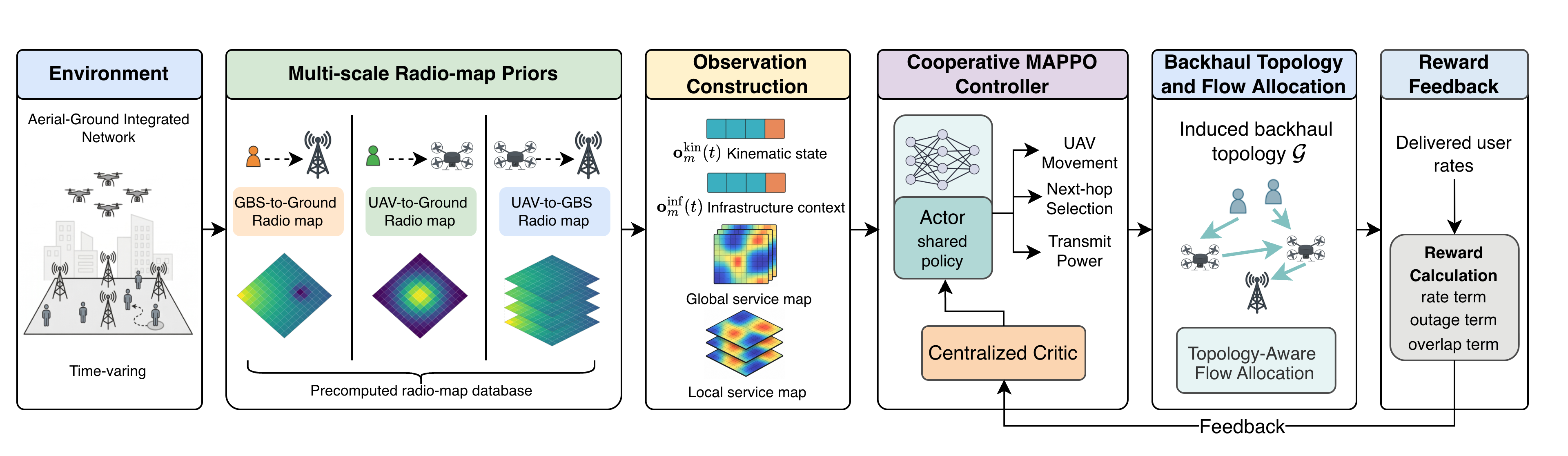}
    \caption{Overview of the proposed MRMG framework. Multi-scale observations constructed from precomputed radio maps are fed into a MAPPO-based cooperative controller, followed by topology-aware flow allocation for E2E service optimization.}
    \label{fig:framework}
\end{figure*}

We develop a radio-map-guided framework for online cooperative multi-\UAV{} control, where GU mobility and urban blockage jointly induce time-varying and location-dependent channel conditions.

\subsection{Multi-Scale Radio Map Construction}
\label{subsec:radiomap}

The proposed controller relies on three types of precomputed radio maps derived from the propagation environment: GBS-to-ground, UAV-to-ground, and GBS-to-UAV. These maps are organized into a global ground-side scale, a local UAV-centered scale, and a pairwise backhaul-link scale. The first two scales are represented as spatial service maps that provide infrastructure-side and local access context, while the third scale is represented as per-\GBS{} backhaul-rate features for candidate \UAV{}--\GBS{} links.

\textbf{1) GBS-to-ground radio map:}
For each ground location $\mathbf{x}$, the map records the strongest GBS-side channel gain,
\begin{equation}
  \bar{g}_{\mathrm{GBS}}(\mathbf{x}) = \max_{b \in \mathcal{B}} g_{b \to \mathbf{x}},
  \label{eq:gbs_radiomap}
\end{equation}
which serves as a global infrastructure-side prior for direct access and \GBS{} reachability. It is further organized into the global service map $\mathbf{o}_m^{\mathrm{glo}}$, which is a coarse $G\times G$ grid consisting of user density, $\bar{g}_{\mathrm{GBS}}(\mathbf{x})$, and a Gaussian position marker at $\mathbf{p}_m(t)$. This map provides a global view for macro-level \UAV{} movement.

\textbf{2) UAV-to-ground radio map:}
For \UAV{} $m$, we query the pre-computed radio-map database at its current position $\mathbf{p}_m(t)$ over a local ground patch centered below the \UAV{}. The resulting large-scale channel gain between \UAV{} $m$ and each ground location $\mathbf{x}$ characterizes the corresponding \GU{}--\UAV{} access channel and is denoted by
\begin{equation}
  \bar{g}_{\mathrm{UAV},m}(\mathbf{x};\mathbf{p}_m(t)),
  \label{eq:uav_radiomap}
\end{equation}
which characterizes nearby access conditions around \UAV{} $m$. This radio information is organized into the local service map $\mathbf{o}_m^{\mathrm{loc}}$, which is a fine-grained $P\times P$ ego-centric grid consisting of user density and $\bar{g}_{\mathrm{UAV},m}(\mathbf{x};\mathbf{p}_m(t))$. This local scale map captures nearby demand and local propagation conditions for positioning and power control.

\textbf{3) GBS-to-UAV Radio Map:}
For each candidate \GBS{} $b$, a precomputed three-dimensional radio map records the channel gain over the discretized \UAV{} flight volume $\mathcal{V}_{\mathrm{air}}$,
\begin{equation}
  \bar{g}_{\mathrm{GBS},b}(\mathbf{p}), \quad \mathbf{p} \in \mathcal{V}_{\mathrm{air}},
  \label{eq:gbs_uav_radiomap}
\end{equation}
where $\mathcal{V}_{\mathrm{air}}$ denotes the set of discrete three-dimensional grid points spanning the \UAV{} altitude range. At the current position $\mathbf{p}_m(t)$, the channel gain to \GBS{} $b$ is read as $g_{m,b}(t)=\bar{g}_{\mathrm{GBS},b}(\mathbf{p}_m(t))$ and used to compute a potential backhaul-rate feature $\tilde{R}_{m,b}(t)$, defined analogously to~\eqref{eq:backhaul_cap} with $P_m(t)$ replaced by $P_{\max}$. Using the maximum transmit power yields a position-dependent upper-bound estimate of the achievable backhaul rate, providing the controller with a channel-quality prior. This feature guides each \UAV{} in selecting between a direct \GBS{} connection and multi-hop forwarding.

\subsection{Policy Design with Radio-Aware Observations}
The joint decisions of multiple \UAV{}s form a cooperative 
multi-agent problem, since each action affects not only the local access conditions of the acting agent but also the shared backhaul topology and A2G interference seen by others. However, a fully centralized controller is impractical for real-time execution. We therefore adopt MAPPO under the centralized-training decentralized-execution (CTDE) paradigm, where a shared critic observes the global network state during training while decentralized actors enable scalable online execution.

\subsubsection{Observation Design}
\label{sec:obs_design}

The local observation of \UAV{} $m$ at slot $t$ is organized into four streams,
\begin{equation}
  \mathbf{o}_m(t)=\bigl(
    \mathbf{o}_m^{\mathrm{kin}}(t),\;
    \mathbf{o}_m^{\mathrm{inf}}(t),\;
    \mathbf{o}_m^{\mathrm{loc}}(t),\;
    \mathbf{o}_m^{\mathrm{glo}}(t)
  \bigr),
  \label{eq:local_obs}
\end{equation}
where $\mathbf{o}_m^{\mathrm{kin}}(t)$ contains the normalized 3-D position of \UAV{} $m$, a slot-state indicator, six boundary-clearance distances, and an agent one-hot identity vector, all obtained from onboard sensing. The infrastructure-context $\mathbf{o}_m^{\mathrm{inf}}(t)$ describes neighboring \UAV{}s and candidate \GBS{} connections; each \GBS{} aggregates the positions of connected \UAV{}s and the user distribution, and broadcasts a summary to all \UAV{}s over the backhaul link. Specifically, each neighboring \UAV{} is encoded by its relative displacement and current slot state, while each \GBS{} $b$ is encoded by its relative displacement together with the potential backhaul rate $\tilde{R}_{m,b}(t)$.

The main spatial context is provided by the two service maps defined in Section~\ref{subsec:radiomap}, both precomputed offline and queried at runtime based on the \UAV{}'s current position. The global map $\mathbf{o}_m^{\mathrm{glo}}(t)$ offers a global view of user density and GBS-side radio conditions to guide macro-level movement, whereas the local map $\mathbf{o}_m^{\mathrm{loc}}(t)$ captures nearby demand and local UAV-to-ground propagation to support fine-grained positioning and power control. Together, these components provide each actor with compact context for action selection.

\subsubsection{Action Space}

At each slot, every \UAV{} jointly selects three actions 
via the actor.
\begin{itemize}
  \item \textit{Movement} $\mathbf{d}_m(t)$: one of seven discrete
    options, including hovering and unit-step movements along each of the three coordinate axes in the positive or negative direction.
  \item \textit{Next-hop selection} $n_m(t)\in(\mathcal{M}\cup\mathcal{B})\setminus\{m\}$:
    the target node for aerial backhaul forwarding, which may be
    another \UAV{} or a \GBS{}.
  \item \textit{Transmit power} $P_m(t)\in\mathcal{P}$:
    selected from a discrete codebook
    $\mathcal{P}=\{0.125,\,0.25,\,0.5,\,1.0\}\cdot P_{\max}$.
\end{itemize}
The actor encodes $\mathbf{o}_m^{\mathrm{kin}}(t)$ and 
$\mathbf{o}_m^{\mathrm{inf}}(t)$ through separate MLPs and each service map through a two-layer CNN; the four streams are fused by an MLP that drives three independent action heads, which sample the three actions simultaneously.

\subsubsection{Reward Design}

The per-agent reward balances three objectives, namely throughput maximization, weak-user protection, and spatial load sharing, and is defined as
\begin{align}
  r_m(t) ={}&
    \frac{\alpha}{R_0}\sum_{k\in\mathcal{K}}\widehat{r}_k(t)
    +\frac{1-\alpha}{R_0}\sum_{k\in\mathcal{K}_m(t)}\widehat{r}_k(t)
    \notag\\
    &-\lambda_{\mathrm{out}}\,\Delta_m^{\mathrm{out}}(t)
    -\lambda_{\mathrm{col}}\,\psi_m(t),
  \label{eq:reward}
\end{align}
where $R_0$ is a rate normalization constant, $\alpha\in[0,1]$ is the team-sharing weight, $\mathcal{K}_m(t)$ is the user set served by \UAV{} $m$, and $\psi_m(t)$ is the spatial-overlap penalty.
The first two terms form a dual-level rate reward. The global term (weighted by $\alpha$) encourages each agent to improve the total network throughput, while the local term (weighted by $1-\alpha$) strengthens individual credit assignment by rewarding the traffic served by \UAV{} $m$ itself. This combination avoids a purely shared reward weakens individual credit assignment.

The third term penalizes weak-user outage based on the rate deficit below the target threshold. We define
\begin{equation}
  \Delta_m^{\mathrm{out}}(t)=
    \Phi^{\mathrm{out}}(t)-(1-\beta)\,\Phi_{-m}^{\mathrm{out}}(t),
  \label{eq:outage_delta}
\end{equation}
where $\Phi^{\mathrm{out}}(t)$ is a smooth rate-shortfall deficit over users below $R_{\min}$, and $\Phi_{-m}^{\mathrm{out}}(t)$ is the counterfactual deficit recomputed with \UAV{} $m$ excluded. Their difference isolates the marginal contribution of agent $m$ to weak-user protection, improving credit assignment while preserving team-level coordination.
The fourth term $\psi_m(t)$ penalizes spatial overlap among \UAV{}s, discouraging multiple agents from collapsing onto the same hotspot and promoting spatial load sharing.

\subsection{Topology-Aware Flow Allocation}

Given the joint actions $\{\mathbf{a}_m(t)\}$, the selected next-hop decisions determine the aerial backhaul topology. We represent this topology as a directed graph $\mathcal{G}(t)=(\mathcal{V},\mathcal{L}(t))$, where the vertex set $\mathcal{V}=\mathcal{M}\cup\mathcal{B}$ includes all \UAV{}s and \GBS{}s, and the link set $\mathcal{L}(t)=\{(m,n_m(t))\mid m\in\mathcal{M}\}$ contains the outgoing backhaul link selected by each \UAV{}. Under the given backhaul topology, the aerially served users share access and multi-hop backhaul resources; hence a topology-aware flow allocation step is required to determine their delivered rates~\cite{zhang2020co}.

Let $\mathcal{K}_{\mathrm{air}}(t)=\bigcup_{m\in\mathcal{M}}\mathcal{K}_m(t)$ denote the aerially served users, and let $\mathcal{P}_k(t)\subseteq\mathcal{L}(t)$ be the selected backhaul path of user~$k$ to a \GBS{}. For each aerially served user, we define the path-capacity weight $w_k(t)=\min_{\ell\in\mathcal{P}_k(t)} C_{\ell}(t)$, which represents the bottleneck capacity of its selected backhaul path. We then maximize a common normalized service ratio $\eta(t)\in[0,1]$, where each aerial flow is required to receive at least $\eta(t)w_k(t)$. This bottleneck-aware formulation prevents the allocation from serving only high-throughput paths while respecting the different capacity scales of multi-hop paths. The resulting weighted max-min flow allocation problem is
\begin{subequations}\label{prob:P1}
\begin{alignat}{2}
  \textbf{P1:}& \max_{\{r_k(t)\},\,\eta(t)} && \eta(t) \notag \\
  & \text{s.t.}
  && r_k(t)\geq w_k(t)\eta(t),
    \;\forall k\in\mathcal{K}_{\mathrm{air}}(t),
    \label{eq:wmm_fairness}\\
  && & 0\leq \eta(t)\leq 1,
    \label{eq:wmm_eta}\\
  && & 0\leq r_k(t)\leq R_{k,m}^{\mathrm{acc}}(t), \notag \\
  && & \forall m\!\in\!\mathcal{M},k\!\in\!\mathcal{K}_m(t),
    \label{eq:wmm_access}\\
  && & \sum_{k:\,\ell\in\mathcal{P}_k(t)} r_k(t)\leq C_\ell(t),
    \;\forall \ell\in\mathcal{L}(t).
    \label{eq:wmm_backhaul}
\end{alignat}
\end{subequations}
The constraints impose normalized fairness, bound $\eta(t)$, and enforce the access and shared-backhaul capacity limits.
\begin{proposition}
The optimal solution to \textbf{P1} is
\begin{equation}
  \eta^*(t) \hspace{-0.5mm}= \hspace{-0.5mm}\min\!\left(1,
\min_{k\in\mathcal{K}_{\mathrm{air}}} \frac{R_{k,m}^{\mathrm{acc}}(t)}{w_k(t)},
    \min_{\ell\in\mathcal{L}} \frac{C_\ell(t)}{\displaystyle\sum_{k:\,\ell\in\mathcal{P}_k} w_k(t)}
  \right).
  \label{eq:eta_opt}
\end{equation}
\end{proposition}
\begin{proof}
Since \eqref{eq:wmm_fairness} must hold with equality at optimum,
we set $r_k(t) = w_k(t)\eta(t)$ for all $k \in \mathcal{K}_{\mathrm{air}}(t)$.
Substituting into the access constraint \eqref{eq:wmm_access} gives
$\eta(t) \leq R_{k,m}^{\mathrm{acc}}(t)/w_k(t)$ for each $k$,
and substituting into the backhaul constraint \eqref{eq:wmm_backhaul} gives
$\eta(t) \leq C_\ell(t)/\sum_{k:\,\ell\in\mathcal{P}_k} w_k(t)$ for each $\ell$.
Together with \eqref{eq:wmm_eta}, maximizing $\eta(t)$ over all feasible 
constraints yields \eqref{eq:eta_opt}.
\end{proof}

The corresponding optimal rate allocation is $r_k^*(t) = w_k(t)\,\eta^*(t)$ for each $k \in \mathcal{K}_{\mathrm{air}}(t)$, which equalizes the normalized service ratio $r_k^*(t)/w_k(t) = \eta^*(t)$ across all aerial users. This approach prevent severe rate degradation for users with weak multi-hop channel.

\section{Simulations}
\label{sec:results}
We evaluate the proposed framework by examining its overall service performance, the contribution of radio-map observations and joint control components, and its scalability with respect to the \UAV{} swarm size. Simulations are conducted in an urban environment, where large-scale channel gains are obtained from precomputed Sionna ray-tracing radio maps~\cite{sionna}. Ground users follow a Gauss--Markov random-walk mobility model. During training, a curriculum-based speed ramp-up is used to expose the policy to gradually increasing user mobility. The main simulation parameters are summarized in Table~\ref{tab:sim_params}.
\begin{table}[t]
  \centering
  \caption{Simulation parameters.}
  \label{tab:sim_params}
  \scriptsize
  \setlength{\tabcolsep}{3pt}
  \begin{tabular}{lc lc}
    \toprule
    Parameter & Value & Parameter & Value \\
    \midrule
    Area & $1000{\times}1000\,\text{m}^{2}$ & Total bandwidth ($W$) & 100\,MHz \\
    No.\ of \UAV{}s ($M$) & 3 & Subbands ($F$) / $B_{\mathrm{sub}}$ & 10 / 10\,MHz \\
    No.\ of \GBS{}s ($N_{\mathrm{B}}$) & 2 & User tx power ($P_{\mathrm{GU}}$) & 100\,mW \\
    No.\ of users ($K$) & 30 & Half-angle ($\Theta$) & $45^{\circ}$ \\
    \UAV{} altitude range & 50--150\,m & Min.\ rate ($R_{\min}$) & 10\,Mbps \\
    \UAV{} move step & 25\,m/slot & Training episodes & 1500 \\
    Carrier frequency & 4.9\,GHz & Steps/episode & 512 \\
    \bottomrule
  \end{tabular}
\end{table}
\begin{figure}[t]
    \vspace*{0.06in}
    \centering
    \includegraphics[width=0.85\columnwidth]{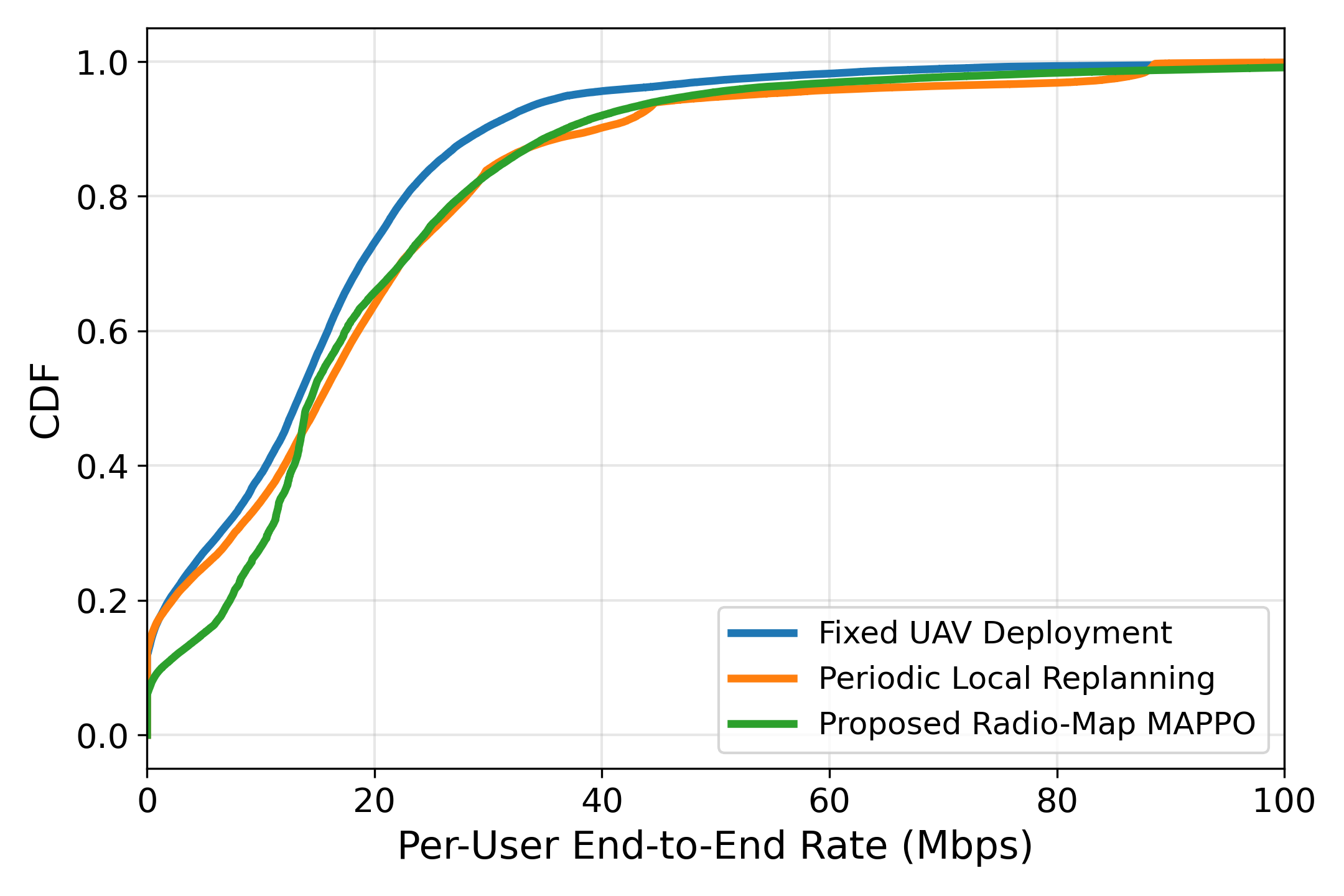}
    \caption{Empirical CDF of per-user end-to-end rates across all methods.}
    \label{fig:cdf}
\end{figure}
\addtolength{\topmargin}{0.06in}
We compare MRMG with four reference schemes spanning different service modes and control strategies.
\begin{itemize}
  \item \textbf{Terrestrial Only}: a ground-only benchmark in which all users are served by \GBS{}s without \UAV{} assistance.
  \item \textbf{Fixed UAV Deployment}: a static backhaul-aware placement baseline~\cite{sabzehali2022optimizing}, where \UAV{} positions are optimized to the initial user distribution and held fixed during execution, while backhaul routing is optimized under the resulting topology.
  \item \textbf{Periodic Local Replanning}: a heuristic control baseline that periodically updates \UAV{} decisions by searching candidate changes within a two-hop local neighborhood.
  \item \textbf{MRMG}: the proposed radio-map-guided MAPPO controller, which learns coordinated multi-\UAV{} decisions.
\end{itemize}

Table~\ref{tab:main_benchmark} reports the main benchmark results. Compared with Periodic Local Replanning, the proposed method increases the coverage ratio from $65.3\%$ to $72.1\%$ and improves the fifth-percentile rate from $0.69$ to $2.70$\,Mbps, while also achieving a higher average rate ($19.03$ vs.\ $18.51$\,Mbps). These gains are more pronounced in coverage and P5 rate than in average rate, indicating that the learned policy better supports users limited by unfavorable access or backhaul conditions. Fig.~\ref{fig:cdf} further illustrates this effect through the empirical CDF of per-user end-to-end rates.
\begin{figure*}[!t]
    \centering
    \includegraphics[width=0.9\textwidth]{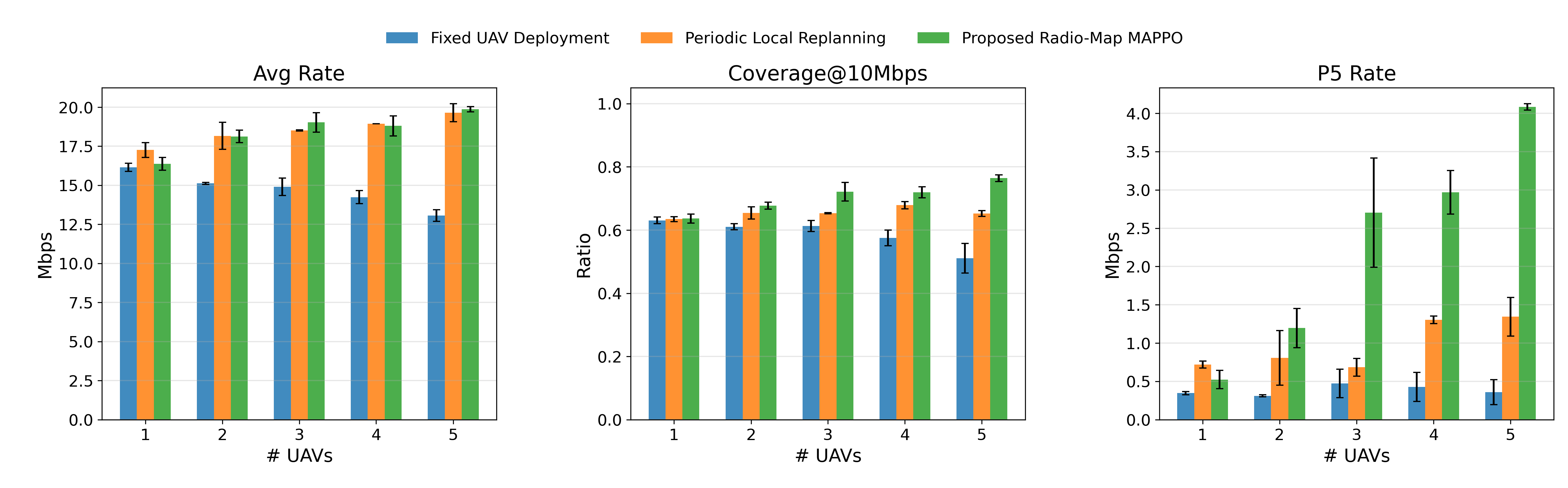}
    \caption{Performance versus number of \UAV{}s ($M=1$ to $5$). MRMG shows increasing advantages in coverage ratio and P5 rate as the swarm grows, while Fixed \UAV{} Deployment degrades in average rate due to misaligned static positioning.}
    \label{fig:uav_count}
\end{figure*}
\begin{table}[t]
    \centering
    \caption{Main benchmark comparison (mean $\pm$ std over 10 episodes).}
    \label{tab:main_benchmark}
    \footnotesize
    \resizebox{\columnwidth}{!}{
    \begin{tabular}{lccc}
        \toprule
        Method & Avg.\ rate & Cov.@10 & P5 rate \\
               & (Mbps)     & (\%)    & (Mbps) \\
        \midrule
        Terrestrial Only          & $15.51{\pm}0.15$ & $56.3{\pm}0.1$ & $0.00{\pm}0.00$ \\
        Fixed UAV Deployment      & $14.91{\pm}0.55$ & $61.3{\pm}1.7$ & $0.47{\pm}0.19$ \\
        Periodic Local Replanning & $18.51{\pm}0.04$ & $65.3{\pm}0.2$ & $0.69{\pm}0.12$ \\
        MRMG                      & $\mathbf{19.03{\pm}0.62}$ & $\mathbf{72.1{\pm}3.0}$ & $\mathbf{2.70{\pm}0.71}$ \\
        \bottomrule
    \end{tabular}
    }
\end{table}
The ablation results are reported in Table~\ref{tab:ablation}. Removing all radio-prior inputs decreases the coverage ratio from $72.1\%$ to $66.3\%$ and reduces the fifth-percentile rate from $2.70$ to $1.17$\,Mbps. Keeping only the backhaul radio prior also leads to lower coverage ($68.9\%$) and weaker low-rate user performance ($1.40$\,Mbps) than the full design, indicating that direct backhaul feasibility should be complemented by global and local service maps for access-and-backhaul control. The movement-and-routing-only and movement-only variants further show that routing and power control contribute to the final performance, especially for improving the lower tail of the user-rate distribution.
\begin{table}[t]
    \centering
    \caption{Ablation study (mean $\pm$ std over 10 episodes).}
    \label{tab:ablation}
    \footnotesize
    \resizebox{\columnwidth}{!}{
    \begin{tabular}{lccc}
        \toprule
        Variant & Avg.\ rate & Cov.@10 & P5 rate \\
                & (Mbps)     & (\%)    & (Mbps) \\
        \midrule
        Full Radio Prior                   & $\mathbf{19.03{\pm}0.62}$ & $\mathbf{72.1{\pm}3.0}$ & $\mathbf{2.70{\pm}0.71}$ \\
        Backhaul Radio Prior Only          & $17.96{\pm}0.94$ & $68.9{\pm}1.0$ & $1.40{\pm}0.00$ \\
        No Radio Prior                     & $17.15{\pm}0.76$ & $66.3{\pm}2.8$ & $1.17{\pm}0.35$ \\
        Movement and Routing Only          & $16.76{\pm}0.13$ & $69.2{\pm}0.1$ & $1.27{\pm}0.05$ \\
        Movement Only                      & $15.82{\pm}0.76$ & $67.2{\pm}1.1$ & $1.13{\pm}0.31$ \\
        \bottomrule
    \end{tabular}}
\end{table}
Fig.~\ref{fig:uav_count} shows performance versus swarm size. The average rate of Fixed \UAV{} Deployment drops as $M$ grows, while both online methods maintain a rising trend.
The advantages of the proposed method in coverage and P5 rate become more evident as $M$ increases. At $M=5$, the proposed method reaches approximately $79\%$ coverage and a $4.1$\,Mbps P5 rate, compared with $67\%$ and $1.3$\,Mbps for Periodic Local Replanning, indicating that the cooperative policy scales more effectively to larger swarms.

\section{Conclusions}
\label{sec:conclusion}
This paper addressed the challenge of online cooperative control in UAV-assisted air-ground integrated networks, where dynamic user mobility and building blockages create significant connectivity bottlenecks. By integrating three-level multi-scale radio maps, the proposed framework successfully decouples regional macro-path planning from transient link-level adaptation. By leveraging this multi-level information architecture, the MAPPO-based controller effectively optimizes joint decisions on UAV movement, routing, and power control under complex 3D constraints. Simulations confirm the superiority of the MRMG framework, which improves service coverage by 6.8 percentage points and doubles the 5th-percentile user rate compared to the existing baselines. These results demonstrate the framework's capability to provide robust, high-quality service to cell-edge users in dynamic environments. Future work will focus on lightweight radio-map reconstruction to enhance the system's autonomy during real-time operations.
\bibliographystyle{IEEEtran}
\bibliography{refs}

\end{document}